\documentclass[runningheads,fleqn]{llncs}
\usepackage{amssymb}
\usepackage{amsmath}
\usepackage{latexsym}
\usepackage{verbatim}
\usepackage{amsfonts}%
\usepackage{graphicx}
\usepackage{todonotes}
\usepackage{latexsym,graphicx}
\usepackage[all]{xy}
\usepackage{xcolor}
\usepackage{url}
\usepackage{verbatim}

\begin{document}
\mainmatter

\title{A note on the avoidability of binary patterns with variables and reversals\thanks{Please note that this work is mostly obsolete as a week ago Currie and Lafrance fully characterised the patterns with reversals in~\cite{CurLaf15}. The only novelty of this work is given by the aperiodicity restriction of the infinite words avoiding such patterns}}

\titlerunning{Palindromes avoidance} 

\author{Robert Merca\c{s}\thanks{Work supported by the P.R.I.M.E. programme of  DAAD co-funded by BMBF and EU's 7th Framework Programme (grant 605728).
}}

\authorrunning{Robert Merca\c{s}}

\institute{Department of Computer Science, Kiel University, Germany,\\
Department of Informatics, King's College London, UK\\
\email{RobertMercas@gmail.com}
}

\maketitle

\begin{abstract}
In this note we present a characterisation of all unary and binary patterns that do not only contain variables, but also reversals of their instances. These types of variables were studied recently in either more general or particular cases. We show that the results are not surprising at all in the general case, and extend the avoidability of these patterns to enforce aperiodic words.
%\footnote{We acknowledge that this is just a note that comes as a completion of the topic, and we intend to extend on this and, hopefully, improve it by the time of publication.}
\end{abstract}

\section{Introduction}

The \emph{pattern unavoidability} concept, was introduced by Bean, Ehrenfeucht and McNulty in~\cite{BeEhMcN} and by Zimin who used the terminology ``blocking sets of terms'' in~\cite{Zim84}. A pattern consisting of variables is said to be unavoidable over a $k$-letter alphabet, if every infinite word over such an alphabet contains an instance of the pattern. That is, there exists a factor of the infinite word which is obtained from the pattern through an assignment of non-empty words to the variables (each occurrence of a variable is substituted with the same word).

The unary patterns, or powers of a single variable $\alpha$, were investigated by Thue~\cite{Thu06,Thu12}: $\alpha$ is unavoidable, $\alpha\alpha$ is 2-unavoidable but 3-avoidable, and $\alpha^m$ with $m\geq 3$ is 2-avoidable. 
Schmidt proved that there are only finitely many binary patterns, or patterns over $E=\{\alpha,\beta\}$, that are 2-unavoidable~\cite{Sch86,Sch89}. Later on, Roth showed that there are no binary patterns of length six or more that are 2-unavoidable~\cite{Rot92}. The classification of unavoidable binary patterns was completed by Cassaigne~\cite{Cas93b} who showed that $\alpha\alpha\beta\beta\alpha$ is 2-avoidable.

After this moment, the concept of unavoidability was investigated in several other context. The ternary patterns were fully characterised in~\cite{Cas94The,Och06}, the binary patterns in the setting of partial words in~\cite{MaMe,BSMeSc09,BSMeRaWi,BSMeSiWe11,FBSLohSco13}, several variations of avoidability of patterns with restrictions on the length of the instances can be found in~\cite{RaShWa05}, while the binary patterns avoidable by cube-free words were characterised in~\cite{MerOchSamShu14} together with their growth rates. However, the topic of our work is mostly inspired by~\cite{RampSha05}, where the authors look at the avoidability of words and their reversals, and~\cite{DuMoScSh14} where the authors show that the pattern $\alpha\alpha\alpha^R$ is avoidable over a binary alphabet, and the work in~\cite{BisCurNow12,ManMueNow12,CurManNow15}, where  a more generalised form of avoidability, that of pseudo-repetitions, is investigated. 
%Even more, there they ask about the growth rate of the binary words avoiding such a pattern, a question settled recently in~\cite{CurRam15}.

In this work, we investigate the avoidability of binary patterns, when some of the variables might be reversed, which, as far as we know, it has not been done before. For example, instead of looking only at the pattern $\alpha\alpha$, we shall also investigate the pattern $\alpha \alpha^R$; this is obviously enough as other variations only consist of complements or mirror images. However, as most of these patterns are avoidable by trivial periodic words (as shown in~\cite{CurLaf15}), we extend a bit our interest and investigate also the cases for which infinite aperiodic words avoiding these patterns exist.

Our work is structured as follows. In the next section we present basic definitions and notations, as well as some preliminary observations. In Section~\ref{sec:unary} we sum up the characterisation for unary patterns, while in Section~\ref{sec:binary} we do the same for binary patterns with reverses.

\section{Definitions and Preliminaries}

Cassaigne's Chapter~3 of \cite{Loth02} provides background on unavoidable patterns, while the handbook itself contains detailed definitions on words. 

Let $\Sigma$ be a non-empty finite set of symbols called an \emph{alphabet}. Each element $0\in \Sigma$ is called a \emph{letter}. A \emph{word} is a sequence of letters from $\Sigma$. The \emph{empty word} is the sequence of length zero, denoted by $\varepsilon$. The set of all finite words~(respectively, non-empty finite words) over $\Sigma$ is denoted by $\Sigma^*$ (respectively, $\Sigma^+$).

A word $u$ is a \emph{factor} of a word $v$ if there exist $x$, $y$ such that $v=xuy$ (the factor $u$ is \emph{proper} if $u\neq \varepsilon$ and $u \neq v$). We say that $u$ is a \emph{prefix} of $v$ if $x=\varepsilon$ and a \emph{suffix} of $v$ if $y=\varepsilon$. The \emph{length} of $u$ is denoted by $|u|$ and represents the number of symbols in $u$. We denote by $u[i..j]$, where $0\leq i \leq j < |u|$, the factor of $u$ starting at position $i$ in $u$ and ending at position $j$, including. By $|u|_v$ we denote the number of distinct, possibly overlapping, occurrences of a factor $v$ in $u$. We denote by $u^R=u[|u|-1]\cdots u[1]u[0]$, the {\em reversal} or {\em mirror image} of a word $u$. A word $u$ is called a {\em palindrome} if $u=u^R$.

For a word $u$, the powers of $u$ are defined recursively by $u^0=\varepsilon$ and for $n\geq1$, $u^n=uu^{n-1}$. Furthermore, $\displaystyle\lim_{n\to\infty}u^n$ is denoted by $u^\omega$. For legibility, the 2-powers of a word are called squares, while 3-powers are called cubes. Furthermore, if $u=v^kv'$, where $v'$ is a prefix of $v$, we say that $u$ is a $\frac{k|v|+|v'|}{|v|}$-power.

Let $E$ be a non-empty finite set of symbols, distinct from $\Sigma$, whose elements are denoted by $\alpha, \beta, \gamma$, etc. Symbols in $E$ are called \emph{variables}, and words in $E^*$ are called {\em patterns}.
The \emph{pattern language} over $\Sigma$ associated with a pattern $p\in E^*$, denoted by $p(\Sigma^+)$, is the subset of $\Sigma^*$ containing all words of $\varphi(p)$, where $\varphi$ is any non-erasing morphism that maps each variable in $E$ to an arbitrary non-empty word from $\Sigma^+$. A word $w\in \Sigma^*$ \emph{meets} the pattern $p$ (or $p$ {\em occurs} in $w$) if for a factorization $w=xuy$, we have $u\in p(\Sigma^+)$. Otherwise, $w$ \emph{avoids} $p$. 

More precisely, let $p = \alpha_0 \cdots \alpha_m$, where $\alpha_i \in E$ for $i\in\{0, \ldots, m\}$.
Define an {\em occurrence} of $p$ in a word $w$ as a factor $u_0 \cdots u_m$ of $w$, where for $i, j \in \{0, \ldots, m\}$, if $\alpha_i = \alpha_j$, then $u_i = u_j$.
Stated differently, for all $i \in \{0, \ldots, m\}$, $u_i \subset \varphi(\alpha_i)$, where $\varphi$ is any non-erasing morphism from $E^*$ to $\Sigma^*$ as described earlier. These definitions extend to infinite words $w$ over $\Sigma$ which are functions from $\mathbb{N}$ to $\Sigma$.

Considering the pattern $p=\alpha\beta\beta\alpha$, the language associated with $p$ over the alphabet $\{0,1\}$ is $p(\{0,1\}^+)=\{uvvu\mid u,v\in\{0,1\}^+$. The word $001100$ meets $p$ (take $\varphi(\alpha)=00$ and $\varphi(\beta)=1$), while the word $01011$ avoids $p$.

Let $p$ and $p'$ be two patterns. If $p'$ meets $p$, then $p$ {\em divides} $p'$, which we denote by $p\mid p'$. For example, $\alpha\alpha \nmid \alpha\beta\alpha$ but $\alpha\alpha \mid \alpha \beta \alpha \beta$. When both $p\mid p'$ and $p'\mid p$ hold, the patterns $p$ and $p'$ are  \emph{equivalent}, and this happens when and only when they differ by a permutation of $E$. For instance, $\alpha\alpha$ and $\beta\beta$ are equivalent.

A pattern $p \in E^*$ is \emph{k-avoidable} if in $\Sigma^*$ there are infinitely many words that avoid $p$, where $\Sigma$ is a size $k$ alphabet. On the other hand, if every long enough word in $\Sigma^*$ meets $p$, then $p$ is $k$-{\em unavoidable} (unavoidable over $\Sigma$).
Finally, a pattern $p \in E^*$ which is $k$-avoidable for some $k$ is simply called \emph{avoidable}, and one which is $k$-unavoidable for every $k$ is called \emph{unavoidable}. The {\em avoidability index} of $p$ is the smallest $k$ such that $p$ is $k$-avoidable, or it is $\infty$ if $p$ is unavoidable. 

In the rest of this work, we only consider binary 
patterns, hence we fix 
$E=\{\alpha,\beta\}$. Moreover, we define $\overline{0} = 1$ and $\overline{1} = 0$, and, similarly, $\overline{\alpha}=\beta$ and $\overline{\beta}=\alpha$, as complementing letters and, respectively, variables.

\paragraph{Preliminaries.} In this paper we are interested in the avoidability of binary patterns in a more general setting. That is, we look at patterns formed not only from variables, but also from reversals. 
As it can be seen, the word $0011001$ has three occurrences of the pattern $\alpha\alpha$, but also has no fewer than six occurrences of the pattern $\alpha\alpha^R$, when $\alpha\in\{0,01, 001,1, 10\}$. Furthermore, it has no occurrence of $\alpha\alpha\alpha$, but one occurrence of $\alpha\alpha^R\alpha$ for $\alpha=01$.

\begin{remark}
A pattern of the form $\alpha\alpha^R$ is equivalent to an even length palindrome. 
\end{remark}

In~\cite{FicZam13},~\cite{DaMaMeMue14} and~\cite{DuMoScSh14} some results regarding the avoidability of palindromes under certain conditions have already been provided.

When considering a four letter alphabet, following a result of Pansiot~\cite{Pan83} we have that there exist infinite words that avoid palindromes. This is due to the fact that over a four letters alphabet there exists an infinite word that has the repetitive threshold $7/5$, thus does not contain any factors of the form $00$ or $010$, for $0,1$ letters, since these would create a $2$, respectively, a $3/2$-power.

When analysing ternary alphabets as to avoid all palindromes, therefore also factors of forms $00$ and $010$ for any letters $0,1$ of the alphabet, we get that the only infinite word that avoids palindromes is $(012)^{\omega}$ and all its suffixes.

For binary alphabets the avoidability of palindromes is not possible as every word of length 3 would contain one.

However, since $\alpha\alpha^R$ is an even length palindrome, the following is immediate:
\begin{remark}
Any square-free word will avoid all even length palindromes.
\end{remark}

Therefore, we already have an upper limit on our avoidability indices.

\section{Unary patterns}\label{sec:unary}
In this section we investigate the avoidability of patterns formed only from a variable and its reversal.

Obviously, when considering a unary alphabet no pattern is avoidable.
The results of Thue~\cite{Thu06,Thu12}, give us precise bounds for the cases when reversals do not occur. Squares are avoidable on a ternary alphabet, while for powers of at least three a binary alphabet is enough.

For the case of reversals, as seen above, a ternary alphabet is enough to avoid any word containing a variable and its reversal, that is a pattern of the form $\alpha\alpha^R$. On further investigation, we see that this is also the case for a binary alphabet, whenever we consider for example the word $(01)^{\omega}$.
Therefore, a first straightforward result is the following:
\begin{remark}\label{rem:unary_avoid}
Every pattern $p$ that has either $\alpha\alpha^R$ or $\alpha^R\alpha$ as factor is $2$-avoidable.
% over a binary alphabet.
\end{remark}

However, both previously given words, $(01)^\omega$ and $(012)^\omega$, are periodic, thus not that interesting. Moreover, any infinite binary or ternary words avoiding such patterns are in fact suffixes of these two words. One direction of our investigation shall deal with the avoidability of these patterns in aperiodic words, e.\,g., words that are not ultimately periodic, thus of the form $uv^{\omega}$. 
%Another direction of our investigation will have as focus the discovery of the lowest sizes of the lowest alphabets 

A first step in this direction was made in~\cite{DuMoScSh14}, where the authors show that the pattern $\alpha\alpha\alpha^R$, which could be depicted in the English word $bepepper$ by taking $\alpha=ep$, is avoidable on a binary alphabet, e.\,g. Theorem~36. Furthermore, the same work conjectures that every binary aperiodic word avoiding this pattern has critical exponent $\geq 2+\varepsilon$, while the number of these words, grows between polynomial and exponential, fact shown recently in~\cite{CurRam15}.

It is also not that difficult to find a binary aperiodic infinite word that avoids the pattern $\alpha\alpha^R\alpha$. For this consider the binary word $\tau=(01)^{\omega}$. Next we ``double'' in $\tau$ a $1$ at positions exponentially far away from the first, and denote the newly obtained word $\tau'$. That is, if we inserted a $1$ at position $k$ in $\tau$ then at position $k-1$ we have a 1, and the next $1$ will be inserted at some position greater than $2k$ after an occurrence of another 1. We have:
$$\tau' = 01 \underline{1} 01 01 \underline{1} 01 01 01 01 \underline{1} 01 01 01 01 01 01 01 01 \underline{1}\cdots,$$
where the new inserted character is underlined.

\begin{lemma}
The word $\tau'$ avoids the pattern $\alpha\alpha^R\alpha$.
\end{lemma}
\begin{proof}
We already know from Remark~\ref{rem:unary_avoid} that the pattern  $p=\alpha\alpha^R\alpha$ is avoidable by $\tau$. Thus in order for $p$ to meet $\tau'$ it must be that one occurrence of it in the word would contain at least one factor $11$. Denote first such occurrence by $uu^Ru$. If $11$ is a factor of $u$, then a contradiction is easily reached given the lengths of the factors (the lengths between each two occurrences of $11$ are increasing exponentially according to the definition). Thus assume now that $11$ occurs only twice in $uu^Ru$, that is $u$ starts and ends in $1$ (11 occurs at the limit between $u$ and $u^R$). However, in this case we have that the length of $u^R$ is double the length from the beginning of the word to the end of the first $u$. To not reach a contradiction with the fact that $u$ contains no 11 as a factor, it must be that $u$ starts at the beginning of the word. A simple check of the prefix of length 10 of $\tau'$ proves that this is not the case. This concludes our proof.
\qed\end{proof}

%Following the above discussion, we have a characterisation of all unary patterns with reverses.
The above discussion fully characterises all unary patterns with reverses.

\begin{theorem}\label{thm:unary_pat}
Let $p\in\{\alpha,\alpha^R\}^+$ be a pattern. Then
\begin{enumerate}
\item
$p$ is unavoidable, whenever $p\in\{\alpha,\alpha^R\}$;
\item
there exist infinite aperiodic ternary words that avoid $p$, whenever $|p|>1$.
%$p$ is avoidable over a ternary alphabet, whenever $|p|>1$;
\item 
$p$ is avoidable over a binary alphabet, whenever $p$ has $\alpha\alpha^R$ or $\alpha^R\alpha$ as factor;
\item
there exist infinite aperiodic binary words that avoid $p$, whenever $|p|>2$.
\end{enumerate}
\end{theorem}

\section{Binary patterns}\label{sec:binary}

%We already established the avoidability indexes of all unary patterns in the previous section. 
In order to start the investigation of binary patterns together with reversals, we have to first recall the results characterising the classical avoidability of binary patterns. For more details see~\cite[Chapter 3]{Loth02}.
%\cite[Chapter 3]{Loth02}
\begin{theorem}
\label{fullbinarypatterns}
Regarding avoidability, binary patterns fall into three categories:
\begin{enumerate}
\item \label{item1}
The binary patterns  $\varepsilon, \alpha, \alpha \beta, \alpha \beta \alpha$, and their complements, are unavoidable (or have avoidability index $\infty$).
\item \label{item2}
The binary patterns $\alpha\alpha$, $\alpha\alpha\beta$, $\alpha\alpha\beta\alpha$, $\alpha\alpha\beta\beta$, $\alpha\beta\alpha\beta$, $\alpha\beta\beta\alpha$, $\alpha\alpha\beta\alpha\alpha$, $\alpha\alpha\beta\alpha\beta$, their reverses, and complements, have avoidability index 3.
\item \label{item3}
All other binary patterns, and in particular all binary patterns of length six or more, have avoidability index 2.
\end{enumerate}
\end{theorem}

 Using further the results of Theorem~\ref{thm:unary_pat}, we shall establish the avoidability of all binary patterns also in the case when reversals are present. 

First, for {\bf item~\ref{item1}} of Theorem~\ref{fullbinarypatterns}, all of the patterns are trivially unavoidable on a unary alphabet, while already any factor of the form $axa$, where $a$ is a letter and $x$ is any non-empty word not containing $a$, meets every one of them.

Now, from item 2 of Theorem~\ref{thm:unary_pat}, we conclude that for all of the patterns at {\bf item~\ref{item2}} of Theorem~\ref{fullbinarypatterns}, except for $\alpha\beta\alpha\beta$,  there exists an aperiodic ternary word avoiding them, no matter how we replace $\alpha$ by $\alpha^R$ or $\beta$ by $\beta^R$. 
%For $\alpha\beta\alpha\beta$ its variations with reversals are $\alpha^R\beta\alpha\beta$ and $\alpha^R\beta^R\alpha\beta$.
%Next, we shall consider the square-free word $\sigma$ obtained from $a$, by infinitely replacing $0$ with $012$, $1$ with $20$, and $2$ with $1$, seee~\cite{Hal64}.
%
%$$\sigma=012021012102012021020121\cdots$$
%
We just have to see now if $3$ is in fact the smallest index possible.

\begin{remark}\label{rem1}
Since $\alpha\alpha^R$ is avoidable by $(01)^\omega$, all patterns  $\alpha\alpha$, $\alpha\alpha\beta$, $\alpha\alpha\beta\beta$, $\beta\alpha\alpha\beta$, and $\alpha\alpha\beta\alpha\alpha$, their reverses, and complements, have avoidability index 2, whenever  one $\alpha$ is replaced by $\alpha^R$. This is also true for $\alpha^R\alpha\beta\alpha$, $\alpha\alpha^R\beta\alpha$, $\alpha\alpha^R\beta\alpha^R\alpha$, $\alpha^R\alpha\beta\alpha^R\alpha$, and all variations of $\alpha\alpha\beta\alpha\beta$ with one of the first two $\alpha$'s reversed.
%, $\alpha\alpha^R\beta^R\alpha\beta$, $\alpha^R\alpha\beta^R\alpha\beta$ and $\alpha^R\alpha\beta\alpha\beta$.
\end{remark}

In this context of avoidability when periodicity is allowed, we still have to analyse $\alpha\alpha\beta\alpha^R$, $\alpha\beta\beta\alpha^R$, $\alpha\alpha\beta\alpha^R\alpha^R$, and the variations of $\alpha\beta\alpha\beta$ and $\alpha\alpha\beta\alpha\beta$.

For the aperiodic case, obviously, as $\beta$ can be chosen to be an arbitrary word, none of the first three patterns of {\bf item~\ref{item2}} of Theorem~\ref{fullbinarypatterns} is avoidable by an aperiodic binary alphabet wherever $\alpha^R$ occurs.

For $\alpha\alpha\beta\alpha$ and $\alpha\alpha\beta\alpha\alpha$, it is immediate that since every infinite binary word contains $00$ or $01010$ as recurring factors (or their complements), they and any of their variations with reverses, other than the ones of Remark~\ref{rem1},  (take $0$ or $01$ as the image of $\alpha$) will occur in every binary infinite word. Thus none of these is avoidable by either an ultimately periodic or aperiodic infinite binary word.  

For $\alpha^R\beta^R\alpha\beta$ see that a word avoiding it should not contain any unary $4$-powers or square of the form $0^i1^j0^i1^j$.

\begin{lemma}\label{lem:avoid1}
The pattern $\alpha^R\beta^R\alpha\beta$ has avoidability index $3$ and there exists an infinite aperiodic ternary word avoiding the pattern.
\end{lemma}
\begin{proof}
Assume this is not the case and there exists a binary infinite word $\gamma$ that avoids it. According to the previous remark $\gamma$ contains no unary $4$-powers. Assume first that $\gamma$ has $10001$ as a factor (the case of $01110$ is symmetrical). Then we can represent a factor starting with this prefix as $10001^i0^j1^k0^\ell$ with $i,j,k,\ell<4$. Observe that in this case, if $i\leq k$, then we can take the image of $\alpha$ to be $0^j$ and that of $\beta$ to be $1^i$, and we would reach a contradiction. Continuing the reasoning, we get in the same manner that $j>\ell$, and so forth we would end up with a factor of the form $1111$ or reach a contradiction. 

We next assume that $\gamma$ has $1001$ as a factor (the case of $0110$ is symmetrical). In this case we get a factor of the form $1001^i0^j1^k0^\ell$, where $i>k$, $2\geq j > \ell=1$, and so forth, reaching once more a contradiction.

Therefore, we can assume that $\gamma$ has $(01)^\omega$ as a suffix. However, in this case choosing the image of $\alpha$ to be $010$ and that of $\beta$ to be $1$, we found an instance of the pattern in the word, which is a contradiction with our initial assumption.

To see that an infinite aperiodic ternary word avoiding the pattern exists, consider the word $\tau''$ obtained from $\tau'$ by insertion of a $2$ before each $0$:
$$\tau'' = 201 \underline{1} 201 201 \underline{1} 201 201 201 201 \underline{1} 201 201 201 201 201 201 201 201 \underline{1}\cdots.$$

First we observe that no length $4$ square exists in this word, thus either the image of $\alpha$ or that of $\beta$ has to have length greater than $1$. However, since the only length greater than $1$ factor of $\tau''$ which also occurs as a reverse is $11$, we conclude that in fact, this image should consist of $11$, while the other should be a single letter. Since $11$ is always preceded by $0$ and followed by $2$ such a scenario is not possible, and thus we conclude that $\tau''$ avoids $\alpha^R\beta^R\alpha\beta$.
\qed\end{proof}

Moreover, it is well known that every infinite binary word contains the square of a word of length greater than $1$, see, e.\,g.,~\cite{RaShWa05}. Thus  $\alpha^R\beta\alpha\beta$ is not avoidable on a binary alphabet, even for ultimately periodic words (take the image of $\alpha$ to have length 1). 
Next we recall a result from~\cite{DaMaMeMue14}.

\begin{theorem}\label{useful}
Over a ternary alphabet there exist infinitely long words that avoid all squares of words with length at least 2 and palindromes of lengths 3 and longer.
\end{theorem}
In particular, the morphism $\psi$, that is defined by
    \begin{align*}
    \hspace{-1.5em}
        \psi(0) &= 011220012201, \qquad
        \psi(1) &= 122001120012, \qquad
        \psi(2) &= 200112201120,
    \end{align*}
maps any infinite square-free ternary word to a word with the desired property. Let us denote such a word obtained by the application of $\psi$ as $\sigma$.

\begin{lemma}\label{lem:aRbab}
The pattern $\alpha^R\beta\alpha\beta$ has avoidability index $3$ and there exists an infinite aperiodic ternary word avoiding the pattern. 
\end{lemma}
\begin{proof}
We claim that $\sigma$ avoids $\alpha^R\beta\alpha\beta$. First observe that since $\sigma$ contains no squares of words of length 2 or longer, it must be that the image of $\alpha$ is greater than $1$. Furthermore, since it contains no palindromes of length $3$ or longer, it must also be that the image of $\beta$ is greater than $1$.

Moreover, we note that $\sigma$ does not contain $10$, $02$, or $21$ as factors. Thus it must be that none of these or their mirror images are factors of the image of $\alpha$. However, it follows in this case that the image of $\alpha$ must be unary, and therefore, have length at most $2$. But this is again a contradiction with the fact that $\sigma$ contains no squares of words of length $2$ or longer.
\qed\end{proof}

For the pattern $\alpha\beta\beta\alpha$, we know it has index 3. 
In every infinite aperiodic binary word we have either  $0110$, $1111$, $10^i1110^i$, or $0^i1110^i1$ as factors, for some $i>0$, or one of their conjugates. 
It immediately follows that  $\alpha\beta\beta^R\alpha$ is unavoidable on a binary alphabet, by an aperiodic infinite word. 
Furthermore, a binary word avoiding $\alpha\beta\beta\alpha^R$ or $\alpha\beta\beta^R\alpha^R$, would have to be of the form $w=\prod 0^{\{1,3\}}1^{\{1,3\}}$.  But since every such aperiodic word contains $101011$ as a factor, we have that $\alpha\beta\beta\alpha^R$ meets every aperiodic infinite binary word. For the word $(0111)^\omega$, observe that the pattern occurs in it as the factor $1101110111$, where $\alpha$ goes to $1$ and $\beta$ to $1011$. However, $\alpha\beta\beta\alpha^R$ does not meet $(01)^\omega$. This is straightforward, as the image of $\beta\beta$ would have even length, and thus would always be preceded and followed by different characters. As the image of $\alpha$ ends with the same letter as the image of $\alpha^R$ begins, the conclusion follows. 

\begin{lemma}\label{lem4}
The pattern $\alpha\beta\beta^R\alpha^R$ has avoidability index $2$ and there exists an infinite aperiodic binary word avoiding the pattern.
\end{lemma}
\begin{proof}
Let us apply a strategy similar to above and triple 1's at positions exponentially apart from the beginning in the word $\tau=(10)^{\omega}$. Furthermore, in order to make later on further use of the constructed word, we shall also impose the condition that between every two consecutive factors $111$ there is an odd number of $0$s. We have the word
$$\tau'''=01 \underline{11} 01 01 01 \underline{11} 01 01 01 01 01 01 01 \underline{11} 01 01 01 01 01 01 01 01 01 01 01 01 01 01 01 \underline{11}\cdots,$$

Observe that in fact, our pattern is an even length palindrome. However, since $\tau'''$ contains none of $00$, $0110$, or $1111$, as a factor, it follows immediately that no even length palindrome of length greater than 3, can exist in $\tau'''$.
%
%Obviously, no factor of length $4$ meets our patterns, as our word does not contain $0000$, $0110$, or a conjugate of them. If the image of $\beta$ starts and ends with the same letter, then this has to be $1$, as this is the only square of a letter that occurs in the word. But then, it must be that either the image of $\beta$ starts with $11$ or ends in it. If it ends in $11$, then the factor $1111$ occurs in our word, which is a contradiction. Hence, it will follow that the image of $\alpha$ ends in $1$, and therefore that of $\alpha^R$ starts with it. However, since the $111$ groups are exponentially apart, we get that the lengths of the images of $\beta$ differ, which is a contradiction. For the same reasons, no factor meeting the pattern can contain more than two occurrences of $111$ as factors. 
%
%Following the same reasoning,  the image of $\beta$ cannot end in $0$, in general. Assume therefore that the factor $uvv^Ru^R$ meets the pattern. Then, if $|v|>1$, it must be that $v$ starts with $0$ and ends in $1$. Hence, $v=(01)^i1$. However, this is a contradiction since $vv^R$ contains $1111$ as a factor. If $v=1$, then we once more get a contradiction as $1111$ would be a factor of the word.
\qed\end{proof}

For this case, we are left to consider the variations of the patterns $\alpha\alpha\beta\beta$ and $\alpha\alpha\beta\alpha\beta$. For the former, we know that when we reverse one variable, the pattern is $2$-avoidable according to item 3 of Theorem~\ref{thm:unary_pat}. Thus we only need to consider this pattern in the context of aperiodic infinite words.
Obviously any variation of the pattern meets every word that has $0011$ or $1111$ as factors.

%\begin{lemma}
%The word $(0111)^\omega$ avoids $\alpha\alpha^R\beta^R\beta$.
%\end{lemma}
%\begin{proof}
%Obviously, the only squares that occur in $\tau=(0111)^\omega$ are $11$ or have length, a multiple of $3$. Thus, the last letter of the images of $\alpha$ and $\beta$ have to be $1$. If the image of any of these has length $1$, then the image of the other, has to have length greater than $1$. However, $\tau$ contains no even length palindromes of such lengths.
%\qed\end{proof}
\begin{lemma}\label{lem5}
The pattern $\alpha\alpha^R\beta^R\beta$ has avoidability index $2$ and there exists an infinite aperiodic binary word avoiding the pattern.
\end{lemma}
\begin{proof}
Let us again consider the word $\tau'''$.

Obviously, the only unary square that occurs in the word is $11$. Thus, the last letter of the images of $\alpha$ and $\beta^R$ have to be $1$. If the image of any of these has length $1$, then the image of the other, has to have length greater than $1$. However, $\tau''$ contains no even length palindromes of length greater than $3$.
\qed\end{proof}

\begin{lemma}
The only infinite binary word avoiding $\alpha\alpha^R\beta\beta$ has $(01)^\omega$ as a suffix.
\end{lemma}
\begin{proof}
Let us consider towards a contradiction that there exists an infinite binary word that avoids the pattern. Obviously the pattern contains no unary $4$-power.
%$0000$.

First assume that this word contains $00$ as a factor (the case when it has $11$ is symmetrical). We consider the first occurrence of $00$ in this word, starting after position $1$; this position is preceded by $1$. It is easy to check that every word starting with $100$ and having length 11 contains an occurrence of the pattern.

Hence, our word has to be ultimately periodic, with $01$ as period. To see that this word avoids our pattern it is straightforwards, as it contains no unary square that would be created by the image of $\alpha$ and its reverse.
\qed\end{proof}

%{\color{red} I think that the only ones from this category still to check are the variations of $\alpha\alpha\beta\alpha\beta$.}

The only patterns from {\bf item~\ref{item2}} of Theorem~\ref{fullbinarypatterns} left are variations of $\alpha\alpha\beta\alpha\beta$.

\begin{lemma}\label{lem:aabab}
The patterns $\alpha\alpha\beta\alpha\beta^R$, $\alpha\alpha\beta\alpha^R\beta$, and $\alpha\alpha\beta\alpha^R\beta^R$ are $2$-avoidable.
\end{lemma}
\begin{proof}
Consider again the word $\tau$. If the image of $\alpha$ starts with $0$, it must end with $1$. Thus the image of $\beta$ must start with $0$. For the first pattern,  the image of $\beta$ ends in $1$, which leads to a contradiction as the image of $\beta^R$ would start with a $0$. For the other patterns the image of $\beta$ must end in $0$, which again leads to a contradiction, as we get that either the image of $\beta$ or that of $\beta^R$ would start with $1$ now. The same strategy works also when the image of $\alpha$ starts with $1$.
\qed\end{proof}

\begin{lemma}
There do not exist infinite aperiodic binary words avoiding any variations of the pattern $\alpha\alpha\beta\alpha\beta$ that include reverses.
\end{lemma}
\begin{proof}
The idea of the proof follows that of Lemma~\ref{lem:avoid1}. We shall only give a sketch as, although the idea is simple, the proof is cumbersome.

We assume that such an infinite aperiodic binary word exists, for one of the patterns. A first observation is that such an aperiodic word must contain $00$ or $11$ as factors. Furthermore, the word cannot contain a unary power of length $5$, as in this case, we take each variable to be represented by a letter. 

Next, we consider the largest $i$ such that $0^i$ is a factor of the word (the same works for $1^i$). Then, we have a factor $0^i1^j0^k1^{\ell}0^m$, such that $i\geq k$ and all powers are less than $5$. If $k=1$ or $k=3$, then it must be that $j>\ell$, as otherwise, for the patterns that have one of the $\beta$s reversed, we can take the image of $\alpha$ to be $0^{\lfloor\frac{i-j}{2}\rfloor+1}$ and that of $\beta$ to be $1^j$ or $1^j0$, respectively. For the other patterns, we just have to also consider the relation between $k$ and $m$, namely that $m<\frac{i-k}{2}$. However, if this happens, we repeat the strategy by taking the $\alpha$ form the group $1^j$ and considering now the next block of $1$'s, following $0^m$. In the end we reach a contradiction with the fact that no unary power greater than $4$ exists.
\qed\end{proof}

%{\bf 
%Check the following variations of  $\alpha\alpha\beta\alpha\beta$ for aperiodic binary avoidability:\\
%$\alpha^R\alpha\beta\alpha\beta$\\
%$\alpha^R\alpha\beta^R\alpha\beta$\\
%$\alpha^R\alpha\beta\alpha^R\beta$\\
%$\alpha^R\alpha\beta^R\alpha^R\beta$\\
%$\alpha\alpha^R\beta\alpha^R\beta$\\
%$\alpha\alpha^R\beta^R\alpha^R\beta$\\
%}

Finally, the only patterns left to be considered are the ones of item~\ref{item3}. 
First, observe that all patterns of length six or more that have only one occurrence of one of the variables, either normal or as a reversal, contain a unary factor of length 3 or longer. Following item 4 of Theorem~\ref{thm:unary_pat}, we conclude that for each of these patterns there exists an aperiodic infinite binary word avoiding them.

Therefore, the only patterns that we still have to consider are the variations of $\alpha\alpha\beta\beta\alpha$, $\alpha\beta\alpha\beta\alpha$, $\alpha\alpha\beta\alpha\beta\beta$, $\alpha\beta\alpha\beta\beta\alpha$, $\alpha\beta\alpha\beta\alpha\beta$, and $\alpha\alpha\beta\beta\alpha\alpha$.

\begin{lemma}\label{lem:aabba}
All variations of the pattern $\alpha\alpha\beta\beta\alpha$ that include reverses have avoidability index 2.
\end{lemma}
\begin{proof}
We claim that the word $\tau$ avoids all variations of the pattern. Since the only squares in $\tau$ are powers of $01$ (for $10$ is similar), having one of the first two $\alpha$'s or one of the $\beta$'s reversed, gives us our result, as we would need a unary square. Thus it must be that the last $\alpha$ is reversed. But, in this case, as the image of $\beta$ would start with $0$ and end in $1$, it must be that the image of the $\alpha^R$ starts also with $0$ and ends with $1$, which is impossible.
\qed\end{proof}

Following the results of Lemma~\ref{lem:aabab}, and Lemma~\ref{lem:aabba}, respectively, due to the division of the patterns, the next observations are straightforward.
\begin{lemma}
All variations of the patterns $\alpha\alpha\beta\alpha\beta\beta$, $\alpha\beta\alpha\beta\beta\alpha$, and $\alpha\alpha\beta\beta\alpha\alpha$ that include reverses have avoidability index 2.
\end{lemma}

Consider now the following 6-uniform morphism $\psi$ and let us recall a result from~\cite{RaShWa05} that comes together with this morphism. 
$$\psi(0) = 011100,\  \psi(1) = 101100,\  \psi(2) = 111000,\ \psi(3) = 110010,\  \psi(4) = 110001.$$

\begin{theorem}\label{thm:RSW} 
There is a square-free  word $w$ that avoids $02$, $03$, $04$, $13$, $14$, $20$, $24$, $30$, $31$, $41$, $42$, $434010$, and, thus, the only squares in $\psi(w)$ are $00$, $11$, $0101$.
\end{theorem}

The following result is a consequence of this construction:

\begin{lemma}\label{lem11}
The pattern $\alpha^R\beta\alpha\beta\alpha$ has avoidability index 2 and there exist infinite aperiodic binary words avoiding it.
\end{lemma}
\begin{proof}
Let us denote the word obtained from Theorem~\ref{thm:RSW} by $\gamma$. It is straightforward that, since $\gamma$ only contains $0101$ as a square of a word of length greater than $1$, this can be the only thing replacing $\beta\alpha\beta\alpha$. However, this implies that the image of $\alpha$ has length $1$, and no instance of the pattern is in the word, as $1010$ is not a valid factor. Furthermore, we note that any occurrence of a variation of the pattern $\alpha\beta\alpha\beta\alpha$ in $\gamma$ would enforce at least one variable having a non-unary image of length greater than $1$.
%
%{\color{red}Other patterns}
\qed\end{proof}

%Finally, note that Dekking proved in~\cite{De79}
%that two letters are sufficient for avoiding abelian fourth powers,
%and three letters suffice for avoiding abelian cubes.
%Thus we conclude that all variations that include reverses of the patterns $(\alpha\beta)^k$, for $k\geq 3$, are avoidable by an aperiodic ternary infinite word.

Since $\alpha\beta\alpha\alpha\beta$ is avoidable, by Theorem~\ref{fullbinarypatterns} and the pattern divisibility property, $\alpha^R \beta \alpha \beta \beta \alpha$ is also avoidable by the same binary infinite aperiodic word.

\begin{lemma}\label{lem12}
There exist infinite aperiodic binary words avoiding the patterns 
$\alpha^R\alpha^R\beta\beta\alpha\alpha$,
$\alpha^R\alpha\beta\beta\alpha^R\alpha$, 
$\alpha^R\alpha\beta\beta\alpha\alpha^R$, 
$\alpha^R\alpha\beta\beta\alpha\alpha$, 
$\alpha\alpha^R\beta\beta\alpha\alpha$, 
$\alpha \alpha^R \beta \alpha \beta \beta$, 
$\alpha^R \alpha \beta \alpha \beta \beta$,
$\alpha^R \alpha \beta \alpha \beta^R \beta$, and
$\alpha \alpha^R \beta \alpha \beta^R \beta$.
\end{lemma}
\begin{proof}
Consider the previously defined word $\tau'''$. We note that the only even length palindrome of it is $11$, while all squares have the form $11$, $(01)^{2\ell}$, $(10)^{2\ell}$, and $11(01)^{2\ell}$, for some positive integer $\ell$.

For the pattern $\alpha^R\alpha^R\beta\beta\alpha\alpha$, if the image of $\alpha$ is $(01)^k$, then that of $\alpha\alpha$ is $(01)^{2k}$ and that of $\alpha^R\alpha^R$ is $(10)^{2k}$. It must be then that $\beta$ goes to $(01)^k$ or $(10)^k$, which do not generate factors of $\tau'''$. If the image of $\alpha$ is $1$, then that of $\beta\beta$ must be $(10)^{2k}$ or $(01)^{2k}$. Since none of these works either, we conclude in this case.

For $\alpha^R\alpha\beta\beta\alpha\alpha^R$ and $\alpha^R\alpha\beta\beta\alpha^R\alpha$, we apply the same strategy as above, but here the square produced by $\beta\beta$ is flanked by the same palindrome of length at least 2. Since $\beta$ cannot go to either $(01)^k$  nor $(10)^k$, the conclusion follows.

For $\alpha^R\alpha\beta\beta\alpha\alpha$ and $\alpha\alpha^R\beta\beta\alpha\alpha$, note that $\alpha$ must go to $1$, as this is the only palindrome of even length. However, $\beta$ cannot go to either $(01)^k$  nor $(10)^k$, since between each factor $11$ there is an odd number of $0$s.

Finally, for the last four patterns, again $\alpha$ must replace $1$. The contradiction is immediate for the last two patterns, where $\beta$ has to also have $1$ as image. For the other two, $\alpha \alpha^R \beta \alpha \beta \beta$ and $\alpha^R \alpha \beta \alpha \beta \beta$, if the image of $\beta$ starts or ends with $1$ we get a contradiction as neither $0110$ nor $1111$ are factors of $\tau'''$, and if it starts and ends with $0$ we reach a contradiction as $00$ is not a valid factor.
\qed\end{proof}

As a consequence of pattern divisibility and Lemma~\ref{lem4} we have the following
\begin{lemma}
There exist infinite aperiodic binary words avoiding 
$\alpha^R\alpha^R\beta^R\beta\alpha\alpha$,
$\alpha\alpha^R\beta^R\beta\alpha\alpha$,
$\alpha \beta \alpha^R \beta^R \beta \alpha$,
$\alpha \beta \alpha \beta^R \beta \alpha^R$,
$\alpha \beta \alpha^R \beta \beta^R \alpha$, 
$\alpha \beta \alpha \beta \beta^R \alpha^R$, and
$\alpha \alpha \beta^R \beta \alpha^R$.
\end{lemma}

As a consequence of pattern divisibility and Lemma~\ref{lem5} we have the following
\begin{lemma}
There exist infinite aperiodic binary words avoiding the patterns
$\alpha^R \alpha \beta^R \beta \alpha$,
$\alpha \alpha^R \beta^R \beta \alpha$,
$\alpha^R\alpha\beta^R\beta\alpha\alpha$,
$\alpha^R\alpha\beta^R\beta\alpha^R\alpha$,
$\alpha^R\alpha\beta^R\beta\alpha\alpha^R$,
\end{lemma}

Following the results in~\cite{CurLaf15} and the result in Lemma~\ref{lem11}, all variations with reversals of the pattern  $\alpha \beta \alpha \beta \alpha$ are avoidad by infinite aperiodic binary words.

%%%% MAIN

\begin{question}
Do there exist infinite aperiodic binary words avoiding the patterns 
$\alpha^R \alpha \beta \beta \alpha$,
$\alpha \alpha^R \beta \beta \alpha$,
$\alpha \alpha \beta \beta \alpha^R$,
$\alpha \alpha \beta \alpha^R \beta \beta$,
$\alpha^R \alpha \beta^R \alpha \beta \beta$,
$\alpha \alpha^R \beta^R \alpha \beta \beta$,
$\alpha \alpha \beta^R \alpha^R \beta \beta$,
$\alpha \beta \alpha^R \beta \beta \alpha$,
$\alpha \beta \alpha \beta \beta \alpha^R$,
$\alpha^R \beta^R \alpha \beta \beta \alpha$, 
$\alpha \beta^R \alpha^R \beta \beta \alpha$,
$\alpha \beta^R \alpha \beta \beta \alpha^R$,
$\alpha^R \beta \alpha \beta^R \beta \alpha$,
$\alpha^R \beta \alpha \beta \beta^R \alpha$,
their reverses, and complements?
\end{question}

We are now ready to state our main result is:

\begin{comment}
\begin{theorem}
Regarding avoidability, the binary patterns that include variables and reverses fall into the following categories:
\begin{enumerate}
\item \label{item1}
All variations of the binary patterns  $\varepsilon, \alpha, \alpha \beta, \alpha \beta \alpha$, and their complements, are unavoidable (or have avoidability index $\infty$).
\item
The binary patterns $XXXX$, their reverses, and complements, have avoidability index 2.
\item
The binary patterns $XXXX$, their reverses, and complements, are avoidable by infinite binary aperiodic words.
\item
The binary patterns $XXXX$, their reverses, and complements, have avoidability index 3.
\item
The binary patterns $XXXX$, their reverses, and complements, are avoidable by infinite ternary aperiodic words.
\item
All other binary patterns have avoidability index 2 and are avoidable by infinite binary aperiodic words.
\end{enumerate}
\end{theorem}
\end{comment}

\begin{theorem}
Let $p$ be a binary pattern with reversal. Then $p$ is either unavoidable, or avoidable by an infinite aperiodic word defined on either a binary or a ternary alphabet.
\end{theorem}

As future work, one of the main things to do next, would be an analysis of the growth functions of the words that avoid all these variations of patters. The best starting point in this direction would be~\cite{CurRam15}, where the authors show that, surprisingly, the growth of the number of words avoiding the pattern $\alpha\alpha\alpha^R$ is between polynomial and exponential. Recently, we have found out from a discussion with James Currie about similar work being done for the pattern $\alpha\alpha^R\alpha$  both in in~\cite{CurRam15b} as well as by Shallit and Du. Again a growth between polynomial and exponential is obtained for the pattern as well.
To this end, we mention that a variety of proving techniques regarding these growth functions is also present in~\cite{MerOchSamShu14}.

\bibliographystyle{abbrv}
\bibliography{references_noEds}
\end{document}